\documentclass[11pt]{article}

\usepackage{setspace}
\usepackage{amssymb,amsmath}
\usepackage{url}
\usepackage{float}

\floatstyle{boxed}
\restylefloat{figure}

\newtheorem{theorem}{Theorem}[section]

\newtheorem{lemma}[theorem]{Lemma}

\newtheorem{claim}{Claim}[section]

\newcommand{\eps}{\varepsilon}

\newcommand{\cm}[1]{}
\newcommand{\comment}[1]{}

\newcommand{\ceil}[1]{\lceil #1 \rceil }

\newcommand{\bx}{\bar{x}}
\newcommand{\bX}{\bar{X}}
\newcommand{\by}{\bar{y}}
\newcommand{\bs}{\bar{s}}
\newcommand{\bb}{\bar{b}}
\newcommand{\bz}{\bar{z}}
\newcommand{\bk}{\bar{k}}
\newcommand{\bxf}{\bar{x}^{f}}
\newcommand{\xf}{x^{f}}
\newcommand{\mOf}{\mO^f}

\newcommand{\byf}{\bar{y}^{f}}
\newcommand{\yf}{y^{f}}
\newcommand{\lr}{L_r}

\newcommand{\bS}{S}

\newcommand{\bK}{K}
\newcommand{\cD}{{\cal D}}
\newcommand{\cE}{{\cal E}}

\newcommand{\cA}{{\cal A}}
\newcommand{\cI}{{\cal I}}

\newcommand{\bL}{\bar{L}}
\newcommand{\bc}{\bar{c}}
\newcommand{\bw}{\bar{w}}
\newcommand{\bW}{\bar{W}}
\newcommand{\bB}{\bar{B}}

\newcommand{\argmax}{\mbox{argmax}}
\newcommand{\mO}{\mathcal{O}}
\newcommand{\mM}{\mathcal{M}}
\newcommand{\SUB}{\mbox{SUB}}
\newcommand{\MC}{\mbox{MC}}

\newcommand{\alphaf}{\alpha_\varphi}
\newcommand{\f}{\varphi}
\newcommand{\tmO}{\tilde{\mathcal{O}}}
\newcommand{\Prb}[1]{{Pr\left[ #1\right] }}
\newcommand{\abs}[1]{ \left| #1 \right |}
 \newenvironment{proof}{\noindent{\bf Proof:}}{
 \hspace*{\fill} $\Box$ \vskip \belowdisplayskip}

\newenvironment{dl_proof}[1]{\noindent{\bf Proof of Lemma #1:}}{
 \hspace*{\fill} $\Box$ \vskip \belowdisplayskip}
 \newenvironment{dl_claim_proof}[1]{\noindent{\bf Proof of Claim #1:}}{
 \hspace*{\fill} $\Box$ \vskip \belowdisplayskip}
\newenvironment{dl_thm_proof}[1]{\noindent{\bf Proof of Theorem #1:}}{
 \hspace*{\fill} $\Box$ \vskip \belowdisplayskip}

 \def\squarebox#1{\hbox to #1{\hfill\vbox to #1{\vfill}}}

\newlength{\tablength}
\newlength{\spacelength}
\newcommand{\tabstar}{\hspace*{\tablength}}
\newcommand{\spacestar}{\hspace*{\spacelength}}
\def\obeytabs{\catcode`\^^I=\active}
{\obeytabs\global\let^^I=\tabstar}
{\obeyspaces\global\let =\spacestar}
\newenvironment{display}{\begingroup\obeylines\obeyspaces\obeytabs}{\endgroup}
\newenvironment{prog}{\begin{display}\parskip0pt\sf}{\end{display}}

\begin{document}
\title{\Large
Approximations for Monotone and Non-monotone Submodular Maximization
with
Knapsack Constraints\footnote{A preliminary version of this paper appeared in the Proceedings of the 20th Annual ACM-SIAM Symposium on Discrete
Algorithms, New York, January 2009.}}
\author{
Ariel Kulik\thanks{Computer Science Department, Technion, Haifa 32000,
Israel. \mbox{E-mail: {\tt ariel.kulik@gmail.com}}} \and Hadas
Shachnai\thanks{Computer Science Department, Technion, Haifa 32000,
Israel. \mbox{E-mail: {\tt hadas@cs.technion.ac.il}}. Work partially supported
by the Technion V.P.R. Fund, by Smoler Research Fund, and by the Ministry of Trade and Industry MAGNET
program through the NEGEV
Consortium (www.negev-initiative.org). }
\and Tami Tamir
\thanks{ School of Computer Science, The
Interdisciplinary Center, Herzliya, Israel.
\mbox{E-mail: {\tt tami@idc.ac.il}}}
}
\date{}

\maketitle

\begin{abstract}

Submodular maximization generalizes many fundamental problems in
discrete optimization, including Max-Cut in directed/undirected
graphs, maximum coverage, maximum facility location and marketing
over social networks.

In this paper we consider the problem of maximizing any submodular
function subject to $d$ knapsack constraints, where $d$ is a fixed
constant. We establish a strong relation between the discrete
problem and its continuous relaxation, obtained through {\em
extension by expectation} of the submodular function. Formally, we
show that, for any non-negative submodular function, an
$\alpha$-approximation algorithm for the continuous relaxation
implies a randomized $(\alpha - \eps)$-approximation algorithm for
the discrete problem. We use this relation to improve the best
known approximation ratio for the problem to $1/4- \eps$, for any
$\eps > 0$, and to obtain a nearly optimal
$(1-e^{-1}-\eps)-$approximation ratio for the monotone case, for
any $\eps>0$.
We further show that the probabilistic domain defined by a
continuous solution can be reduced to yield a polynomial size
domain, given an oracle for the extension by expectation. This
leads to a deterministic version of our technique.
\end{abstract}
\bigskip

%
%

\section{Introduction} \label{sec:intro}

 A real-valued function $f$, whose domain is all the
 subsets of a universe $U$, is called {\em submodular} if,
 for any $S,T \subseteq U$,
 \[
 f(S) + f(T) \geq f(S \cup T) + f(S \cap T).
 \]
The concept of submodularity, which can be viewed as a discrete analog of convexity,
plays a central role in combinatorial theorems and algorithms
(see, e.g., \cite{f99}
and the references therein, and the comprehensive
surveys in \cite{FMV07,Vo08,LMNS09}).
Submodular maximization generalizes many fundamental problems in
discrete optimization, including Max-Cut in directed/undirected
graphs, maximum coverage, maximum facility location and marketing
over social networks (see, e.g., \cite{HMS08}).

In many settings, including set covering or matroid optimization, the
underlying submodular functions are monotone, meaning that $f(S)
\leq f(T)$ whenever $S \subseteq T$. In other settings, the function $f(S)$
is not necessarily monotone.
A classic
example of such a submodular function is $f(S)= \sum_{e \in \delta(S)} w(e)$,
where $\delta(S)$ is a cut in a graph (or hypergraph) $G=(V,E)$
induced by a set of vertices $S \subseteq V$, and $w(e)$ is the weight of an edge $e \subseteq E$.
An example for a monotone submodular function is $f_{G,\bar{p}}:2^L \rightarrow \mathbb{R}$, defined on a subset of vertices in bipartite graph
$G=(L,R,E)$. For any $S \subseteq V$, $f_{G,\bar{p}}(S)= \sum_{v \in N(S)} p_v$, where $N(S)$ is the neighborhood function (i.e., $N(S)$ is the
set of neighbors of $S$), and $p_v \geq 0$ is the profit of $v$, for any $v \in R$.
The problem
$\max\{f_{G,\bar{p}}(S) | \abs{S} \leq k\}$ is classical maximum coverage.

In this paper we consider the following
problem of maximizing a non-negative {\em submodular}
set function subject to $d$
{\em knapsack constraints} ({$\SUB$). Given a $d$-dimensional
budget vector $\bL$, for some $d \geq 1$, and an oracle for a
non-negative submodular set function $f$ over a universe $U$,
where each element $i \in U$ is associated with a $d$-dimensional
cost vector $\bc(i)$, we seek a subset of elements $S \subseteq U$
whose total cost is at most $\bL$, such that $f(S)$ is maximized.

There has been extensive work on maximizing submodular {\em monotone}
functions subject to matroid constraint.\footnote{A (weighted)
matroid is a system of `independent subsets' of a universe, which
satisfies certain {\em hereditary} and {\em exchange} properties
\cite{Sch03}.} For the special case of uniform matroid,
i.e., the problem
$\{ \max f(S): |S| \leq k \}$, for some $k >1$, Nemhauser et. al
showed in \cite{NWF78}
that a greedy algorithm yields a ratio of $1-e^{-1}$ to the
optimum. Later works presented greedy algorithms that achieve this
ratio for other special matroids or for
variants of maximum coverage (see, e.g., \cite{as04,
kmn99,s04,CK04}).
For a general matroid constraint, Calinescu et al. showed in
\cite{ccpv07} that a scheme based on solving a continuous
relaxation of the problem followed by {\em pipage rounding} (a
technique introduced by Ageev and Sviridenko \cite{as04}) achieves
the ratio of $1-e^{-1}$ for maximizing submodular monotone
functions that can be expressed as a sum of weighted rank
functions of matroids. Subsequently, this result was extended by
Vondr\'{a}k \cite{Vo08} to general monotone submodular functions.

 The bound of $1 - e^{-1}$ is the best possible for all of the
above problems. This follows from the lower bound of Nemhauser and
Wolsey \cite{nw78} in the oracle model, and the later result of
Feige \cite{f98} for the specific case of maximum coverage, under
the assumption that $P \neq NP$.

Other variants of monotone submodular optimization were also considered.
In~\cite{BKNS10}, Bansal et al. studied the problem of maximizing a monotone submodular function
subject to $n$ knapsack constraints, for arbitrary $n \geq 1$, where each element appears in up to $k$
constraints, and $k$ is fixed. The paper presents a $\frac{8ek}{e-1}$ and
$\frac{e^2k}{e-1} +o(k)$ approximations for
this problem. Demaine and Zadimoghaddam \cite{DZ10} studied
bi-criteria approximations for monotone submodular set function
optimization.

The problem of maximizing a {\em non-monotone} submodular function has been studied as well. Feige et al. \cite{FMV07}
 considered (unconstrained) maximization of a  general non-monotone submodular function.
 The paper gives several (randomized and deterministic)
approximation algorithms, as well as hardness results, also for the special case
where the function is {\em symmetric}.

Lee et al. \cite{LMNS09} studied the problem of
maximizing a general submodular function under linear and matroid
constraints. They proposed algorithms that achieve approximation
ratio of $1/5 - \eps$ for  the problem with $d$ linear constraints
and a ratio of $1/(d+2+ 1/d + \eps)$ for $d$ matroid constraints,
for any fixed integer $d \geq 1$.

Improved lower and upper bounds for non-constrained and constrained submodular maximization
were recently derived by Gharan and Vondr\'{a}k
\cite{SV10}. However, this paper does not consider knapsack constraints.

\comment{ Submodular maximization problem where the focus of many
works. The best approximation ratio for maximizing non-monotone
submodular function without any constraints is $2/5-o(1)$ general
functions, and $1/2-\eps$ for symmetric function, both algorithm
are due to \cite{FMV07}. { {\bf write a survey of results for
constraints maximization}. For monotone functions mention the
classic greedy algorithm for cardinality and knapsack constraint.
Later matroid constraints, and afterwards our results. For
non-monotone mention Sviridenko. } }

Several fundamental algorithms for submodular maximization
(see, e.g., \cite{as04,ccpv07,Vo08,LMNS09})
use a continuous extension of submodular function, to which we
refer as \emph{extension by expectation}. Given a submodular
function $f: 2^U \rightarrow \mathbb{R}$, we define $F: [0,1]^U
\rightarrow \mathbb{R}$.
For any $\bar{y} \in [0,1]^U$, let $R \subseteq U$  be a random
variable such that $i\in R$ with probability $y_i$ (we say that
$R\sim \bar{y}$). Then
\begin{equation*}
F(\bar{y})= E[f(R)] =
\sum_{R \subseteq U} \left( f(R) \prod_{i\in R} y_i \prod_{i \notin
R} (1-y_i) \right).
\end{equation*}
%
%
%
The general framework of these algorithms is to obtain first a
fractional solution for the continuous extension, followed by
rounding which yields a solution for the discrete problem.

Using the definition of $F$, we define the continuous relaxation of our problem called
{\em continuous \SUB}. Let $P=\{\bar{y} \in [0,1]^U |
\sum_{i\in U} y_i{\bc}(i) \leq \bL\}$ be the polytope of the instance, then the problem
is to find $\bar{y}\in P$ for which $F(\bar{y})$ is
maximized. For $\alpha \in (0,1]$, an algorithm $\cA$ yields $\alpha$-approximation
 for the continuous problem with respect to a submodular
function $f$, if for any assignment of non-negative costs to the elements, and for any non-negative budget, $\cA$ finds a feasible solution for
continuous ${\SUB}$ of value at least $\alpha \mO$, where $\mO$ is the value of an optimal
(integral) solution for ${\SUB}$ with the given costs and budget.

For some specific families of submodular functions,
linear programming can be used to derive such approximation algorithms
(see e.g  \cite{as04,ccpv07}). For monotone submodular functions,
Vondr\'{a}k presented in \cite{Vo08} a
$(1-e^{-1}-o(1))$-approximation algorithm for the continuous
problem. Subsequently,
Lee et al. \cite{LMNS09} considered the problem of maximizing {\em any}
submodular function with multiple knapsack constraints and developed
a $(\frac{1}{4}-o(1))$-approximation algorithm for the continuous
problem; however, noting that the rounding method of \cite{KST09},\footnote{
The paper \cite{KST09} is a preliminary version of this paper.}
which proved useful for monotone functions, cannot be applied in the non-monotone case,
a $(\frac{1}{5}-\eps)$-approximation was obtained for the discrete
problem, by using simple randomized rounding.
This gap of approximation ratio between the continuous and the
discrete case led us to further develop the technique in \cite{KST09},
so that it can be applied also for non-monotone functions.

\subsection{Our Results}
In this paper
we establish a strong relation
between the problem of maximizing any submodular function subject
to $d$ knapsack constraints and its continuous relaxation.
Formally, we show (in Theorem~\ref{thm:continuous_eq}) that for
any non-negative submodular function, an $\alpha$-approximation
algorithm for the continuous relaxation implies a randomized
$(\alpha - \eps)$-approximation algorithm for the discrete
problem. We use this relation to
obtain approximation ratio of $1/4- \eps$ for {\SUB}, for any
$\eps > 0$, thus improving the best known result for the problem,
due to Lee et al. \cite{LMNS09}. For the case where the objective
function is monotone, we use this relation to obtain a nearly
optimal $(1-e^{-1}-\eps)$ approximation, for any $\eps>0$. An
important consequence of the above relation is that for any
class of submodular functions,
a future improvement of the approximation ratio for the continuous
problem, to a factor of $\alpha$,
immediately implies an approximation ratio of $(\alpha -\eps)$ for
the original instance.

Our technique applies random sampling on the solution space, using
a distribution defined by the fractional solution for the problem.
In Section \ref{mcmb:deter} we show how to convert a feasible
solution for the continuous problem to another feasible solution
with up to $O(\log |U|)$ fractional entries, given an oracle to
the extension by expectation. This facilitates the usage of
exhaustive search instead of sampling, which leads to a
deterministic version of our technique. Specifically, we obtain a
deterministic $(1/4 -\eps)$-approximation for general instances
and $(1- e^{-1} -\eps)$-approximation for instances where the
submodular function is monotone.
For the special case of maximum
coverage with $d$ knapsack constraints,
that is, $\SUB$ where the objective
function is $f=f_{G,\bar{p}}$  for a given bipartite graph $G$ and
profits $\bar{p}$, this result leads to a deterministic $(1-
e^{-1} -\eps)-$approximation algorithm, since the extension by
expectation of $f_{G,\bar{p}}$ can be deterministically evaluated.
Some basic properties of submodular functions are given in
Appendix \ref{app:basic_props}.

\subsection{Recent Developments}

Subsequent to our study of maximizing monotone submodular
functions subject to multiple knapsack constraints \cite{KST09},
Chekuri et al. \cite{CVZ10}
showed that, by using a more
sophisticated rounding technique,
the algorithm in
\cite{KST09} can be applied to derive a
$(1-e^{-1}-\eps)$-approximation for
maximizing a submodular function subject to $d$ knapsack
constraints and a matroid constraint.
Specifically, given a fractional solution for the problem, the
authors
define a probability distribution over the solution space, such
that all of elements in the domain of the distribution are inside
the matroid; these elements also satisfy Chernoff-type
concentration bounds, which can be used to prove some of the
probabilistic claims in \cite{KST09}.
The desired approximation ratio is obtained by using the algorithm
of \cite{KST09} with sampling replaced by the above distribution
in the rounding step.
Recently, the same set of authors improved in \cite{CVZ10a} the
bound of $(1/4- \eps)$ presented here to $0.325$.

\comment{
 In~\cite{CVZ10} Chekuri, Vondr\'{a}k and Zenklusen
presented a $(1-e^{-1}-\eps)$-approximation algorithm for monotone
maximization subject to $d$ knapsack constraints and a matroid
constraint. Their main tool was to provide a distribution (given a
fractional solution), such that every element in the
distribution's domain is inside the matroid, and also maintains a
Chernoff-type concentration bounds which can be used to prove our
probabilistic claims. By applying this distribution over our
algorithm (while referring to \cite{KST09}), instead of a
sampling, the desired approximation ratio was attained. }

 \section{Maximizing Submodular Functions}

\label{sec:randomized_mcmb}

In this section we describe our framework for maximizing a
submodular set function subject to multiple linear constraints.
For short, we call this problem $\SUB$.

\subsection{Preliminaries}
\label{sec:prel}

\paragraph{Notation:}
An essential component in our framework is the distinction between
elements by their costs. We say that an element $i$
is \emph{small} if $\bc (i) \leq \eps^3 \bL$; otherwise,
the element is \emph{big}.


Given a universe $U$, we call a subset of elements $S\subseteq U$
{\em feasible} if the total cost of elements in $S$ is bounded by
$\bL$. We say that  $S$ is \emph{$\eps$-nearly feasible} (or
\emph{nearly feasible}, if $\eps$ is known from the context) if the total
cost of the elements in $S$ is bounded by $(1+\eps) \bL$. We refer
to $f(S)$ as the value of $S$. Similar to the discrete case,
$\bar{y}\in [0,1]^U$ is feasible if $\bar{y} \in P$.

For any subset $T \subseteq U$, we define $f_T: 2^{U} \rightarrow
\mathbb{R}_{+}$ by $f_T(S) = f(S \cup T) - f(T)$. It is easy to
verify that if $f$ is a submodular set function then $f_T$ is also
a submodular set function. Finally, for any set $S \subseteq U$,
we define $c_r(S)=\sum_{i \in S} c_r(i)$, where $1 \leq r \leq d$, and
 $\bc(S)= \sum_{i \in S} \bc(i)$.  For a fractional solution $\by \in [0,1]^U$, we
define $c_r(\bar{y}) = \sum_{i\in U} c_r(i) \cdot y_i$ and
$\bc(\bar{y})= \sum_{i\in U} \bc(i)\cdot y_i$.

\paragraph{Overview:}
Our algorithm consists of two phases, to which we refer as
{\em rounding procedure} and {\em profit enumeration}. The rounding procedure
yields an $(\alpha-O(\eps))$-approximation for
instances in which there are no big elements, using an
$\alpha$-approximate solution for the continuous problem. It relies heavily
on Theorem \ref{thm:prb_claim} that gives some conditions on the
probabilistic domain of solutions; these conditions guarantee that the
expected profit of the resulting nearly feasible solution is high. This
solution is then converted to a feasible one, by using a
fixing procedure. We first present a randomized version
and later show
how to derandomize the rounding procedure.

The profit enumeration phase uses enumeration over the most
profitable elements in an optimal solution; then it reduces
a general instance to another instance with no big elements,
on which we apply the rounding procedure.

Finally, we combine the above results with an algorithm for the
continuous problem (e.g., the algorithm of \cite{Vo08}, or
\cite{LMNS09}) to obtain approximation algorithm for $\SUB$.

\subsection{A Probabilistic Theorem}
\label{sec:prb_claim}

We first prove a general probabilistic
theorem which refers to a slight generalization of our problem
(called {\em generalized $\SUB$}).
In addition to the standard
input for the problem, there is also a collection of subsets $\mM\subseteq 2^U$,
such that if $T\in \mM$ and $S \subseteq T$
then $S\in \mM$. The goal is  to find
a subset $S \subseteq \mM$, such that  $\bc(S) \leq \bL$
and $f(S)$ is maximized.

\begin{theorem}
\label{thm:prb_claim}
For a given input of generalized $\SUB$,
let $\chi$ be a distribution over $\mM$
and $D$ a random variable $D \sim \chi$, such that
\begin{enumerate}
\item $E \left[ f(D) \right] \geq \mO/5 $, where $\mO$ is an optimal
solution for the given instance.
\item For any $1\leq r \leq d$, $E[c_r(D) ]  \leq L_r$
\item For any $1 \leq r \leq d$, $c_r(D)= \sum_{k=1}^{m} c_r(D_k)$, where
 $D_k \sim \chi_k$ and $D_1, \ldots , D_m$ are independent random
variables.
\item
\label{thm:prb_claim:c4}
For any $1\leq k  \leq m$ and $1\leq r \leq d$, it holds that
either $c_r(D_k) \leq \eps ^3 L_r$ or $c_r(D_k)$ is fixed.
\end{enumerate}
Let $D'=D$ if $D$ is $\eps$-nearly feasible, and $D'=\emptyset$ otherwise. Then $D'$ is
always $\eps$-nearly feasible, $D' \in \mM$, and
$E[f(D')] \geq (1-O(\eps) ) E[f(D)]$.
\end{theorem}

To prove the results in this section, it suffices to use a special case of Theorem
 \ref{thm:prb_claim} (formulated as our next result). We use
this theorem in its full generality in \cite{full}, in developing approximation
 algorithms for variants of maximum coverage and GAP.

\begin{lemma}
\label{lemma:expected_random}
Let $\bx \in [0,1]^U$ be a feasible fractional solution
such that
$F(\bx) \geq \mO /5$, where $\mO$ is the optimal solution for generalized ${\SUB}$.
 Let $D \subseteq U$ be a random set such
that $D \sim \bx$ (i.e., for all $i \in U$, $i \in R$ with probability $x_i$),
 and
let $D'$ be a random set such that $D'=D$ if $D$ is $ \eps$-nearly feasible, and $D'=\emptyset$ otherwise.
Then $D'$ is always $\eps$-nearly feasible, and
$E[f(D')] \geq (1-O(\eps) ) F(\bx)$.
\end{lemma}

\begin{dl_thm_proof}{\ref{thm:prb_claim}}
Define an indicator random variable $F$ such that $F=1$ if $D$ is $\eps$-nearly
feasible, and $F=0$ otherwise.

\begin{claim}
\label{lemma:prob_F}
$Pr[F=0] \leq d \eps$.
\end{claim}
\begin{proof}
For any dimension $ 1 \leq r \leq d$,
it holds that $E[ c_r(D) ] = \sum_{k=1}^{m} E[c_r(D_k)] \leq L_r$.
Define $V_r=\{k | c_r(D_k) \mbox{~is not fixed} \}$. Then,
\[
\begin{array}{ll}
Var[c_r(D)] & = \displaystyle{\sum_{k=1}^{m} Var[c_r(D_k)] \leq \sum_{k\in V_r} E[c_r^2(D_k)]}
\\  & \leq \displaystyle{
\sum_{k\in V_r} E [c_r(D_k)] \cdot \eps^3 L_r \leq  \eps^3 L_r \sum_{k=1}^{m} E [c_r(D_k)] \leq
\eps^3 L^2_r}.
\end{array}
\]
The first inequality holds since $Var[X] \leq E[X^2]$, and the second
inequality follows from the fact that $c_r(D_k) \leq \eps^3 L_r$ for
$k\in V_r$.
Recall that, by the Chebyshev-Cantelli inequality, for any $t > 0$ and a random
variable $Z$,
$$Pr\left[Z- E[Z] \geq t\right] \leq \frac{Var[Z]}{Var[Z]+ t^2}.$$
Thus,
\begin{eqnarray*}
Pr\left[c_r(D) \geq (1+\eps) L_r\right] &=&
Pr\left[ c_r(D) - E[c_r(D)] \geq (1+\eps) L_r -E[c_r(D)] \right]  \\
&\leq& Pr \left[ c_r(D) - E[c_r(D)] \geq \eps \cdot  L_r \right]
\leq \frac{ \eps^3 L^2_r}{\eps^2 L^2_r} = \eps.
\end{eqnarray*}
By the union bound, we have that
$$ Pr[F=0] \leq \sum_{r=1}^{d} Pr[c_r(D) \geq (1+\eps) L_r] \leq
d \eps.$$
\end{proof}

For any dimension $1 \leq r \leq d$, let $R_r = \frac{c_r(D)}{\lr}$,
and define $R= \max_r R_r$, then
$R$ denotes the maximal relative deviation of the cost from the $r$-th entry
in the budget vector, where the maximum is taken over $1 \leq r \leq d$.

\begin{claim}
\label{lemma:prob_l} For any $\ell> 1$,
$$Pr[R > \ell] <\frac{d\eps^3}{(\ell-1)^2}.$$
\end{claim}
\begin{proof}
By the Chebyshev-Cantelli inequality we have that, for any dimension
$1 \leq r \leq d$,
\begin{eqnarray*}
Pr[R_r > \ell ] &=&
Pr[c_r(D) > \ell \cdot \lr]
 \\
&\leq&
Pr \left[ c_r(D) -E[c_r(D)]  > (\ell-1) \lr \right] \leq \\
&\leq &
\frac{\eps^3 \lr^2}{(\ell-1)^2 \lr^2} \leq
\frac{\eps^3}{(\ell-1)^2},
\end{eqnarray*}
and by the union bound, we get that
$$Pr[R > \ell] \leq \frac{d\eps^3}{(\ell-1)^2}.$$
\end{proof}

\begin{claim}
\label{lemma:bounding_via_R}
For any integer $\ell >1$, if $R \leq \ell$ then
$$f(D) \leq 2d \ell \cdot \mO.$$
\end{claim}

\begin{proof}
The set $D$ can be partitioned to $2 d \ell$ sets
$D_1, \ldots D_{2 d \ell}$ such that each of these
sets is a feasible solution. Hence, $f(D_i) \leq \mO$.
By Lemma~\ref{lemma:submodular_summation}, we have that $f(D) \leq f(D_1) +\ldots +f(D_{2 d \ell})
\leq 2 d \ell \mO$.
\end{proof}
Combining the above results we have

\begin{claim}
\label{lemma:mcmb_nearly_feasible}
$E[f(D')]\geq (1-O(\eps))
E[f(D)]$.
\end{claim}

\begin{proof}
By Claims \ref{lemma:prob_F} and \ref{lemma:prob_l}, we have that
\begin{eqnarray*}
E[f(D)]    &=&
E\left[f(D) | ~ F=1\right] \cdot \Prb{F=1} +
 E\left[f(D) |~ F=0 \land (R < 2) \right]\cdot \Prb{F=0 \land (R   < 2)} \\
   &+& \sum_{\ell \geq 1}  E\left[f(D)|~  F=0 \land (2 ^\ell \leq R <  2^{\ell+1})\right]
        \cdot
         \Prb{F=0 \land (2 ^\ell \leq R < 2^{\ell+1})} \\
 &\leq& E[f(D) |~ F=1] \cdot \Prb{F=1} + 4d^2 \eps \cdot \mO
  + ~d^2 \eps ^3 \cdot \mO \cdot \sum_{\ell \geq 1}  \frac{2^{\ell+2}}
      { (2^{\ell-1})^2 }. \\
 \end{eqnarray*}
Since the last summation is a constant, and $E[f(D)]\geq \mO/2$, we have that
 $$E[F(D)] \leq E[f(D) | F=1 ] \Prb{F=1} + \eps \cdot c \cdot E[F(D)],$$
where $c>0$ is some constant. It follows that
 $$(1-O(\eps)) E[f(D)] \leq E[f(D) | F=1 ] \cdot \Prb{F=1}.$$
Finally, since $D' = D$ if $F=1$ and $D' =0$ otherwise, we have that
$$E[f(D')] = E[f(D)|F=1] \cdot \Prb{F=1} \geq (1- O(\eps)) E[f(D)]. $$
 \end{proof}
By definition, $D'$ is always $\eps$-nearly feasible, and $D'\in \mM$. This
completes the proof of the theorem.
\end{dl_thm_proof}

\subsection{Rounding Instances with No Big Elements}
\label{sec:random_mcmb}
In this section we present an $(\alpha - O(\eps))$-approximation
algorithm for ${\SUB}$ inputs with no big elements, given an
 $\alpha$-approximate solution for the continuous problem.
 Inputs with no big elements are easier to tackle. Indeed, any nearly feasible solution
 for such input can be converted to a feasible one, with only
 a small harm to the total value.

\begin{lemma}
\label{lemma:nearly_fix}
Let $S\subseteq U$ be an $\eps$-nearly feasible solution
with no big elements,
then $S$ can
be converted in polynomial time to a feasible solution $S' \subseteq S$,
such that $f(S') \geq \left(1-O(\eps) \right) f(S)$.
\end{lemma}

\begin{proof}
In fixing the solution $S$ we handle each dimension separately.
For any dimension $ 1\leq r \leq d$, if $c_r(S)\leq L_r$  then
no modification is needed; otherwise, $c_r(S) > L_r$. Since all elements
in $S$ are small, we can partition $S$ into $\ell$ disjoint subsets $S_1,S_2,\ldots,S_\ell$
such that $\eps L_r \leq c_r(S_j) < (\eps +\eps^3) L_r$ for any
 $1\leq j \leq \ell$, where $\ell = \Omega( \eps ^{-1})$.
Since the function $f$ is submodular, by Lemma \ref{lemma:submodular_summation_inequality},
 we have that
$ f(S) \geq \sum_{j=1}^{\ell} f_{S \setminus S_j} (S_j)$. Hence, there exists a value $j \in \{ 1, 2 \ldots , \ell \}$ such that $f_{S \setminus
S_j} (S_j) \leq \frac{f(S)}{\ell}= f(S) \cdot O(\eps)$ (note
that $f_{S \setminus S_j} (S_j)$ may be negative). Now, $c_r(S \setminus S_j) \leq
L_r$, and $f(S \setminus S_j) \geq (1-O(\eps)) f(S)$. We repeat this step for all $1 \leq r \leq d$ to obtain a feasible set $S'$ satisfying
$f(S') \geq (1-O(\eps)) f(S)$.
\end{proof}
Combined with Theorem \ref{thm:prb_claim},
we have the following rounding algorithm. \\ \\
\noindent
{\bf Randomized Rounding Algorithm for $\SUB$ with No Big Elements} \\
\noindent
{\bf Input:} A $\SUB$ instance, a feasible solution  $\bar{x}$ for the continuous problem, with
$F(\bar{x}) \geq \mO /5$.
\begin{enumerate}
\item Define a random set $D \sim \bx$. Let $D'=D$ if
$D$ is $\eps$-nearly feasible, and $D'=\emptyset$ otherwise.
\item Convert $D'$ to a feasible
set $D''$ as in the proof of Lemma \ref{lemma:nearly_fix} and return $D''$.
\end{enumerate}

Clearly, the algorithm returns a feasible solution for the problem.
By Theorem \ref{thm:prb_claim}, $E[f(D')] \geq (1- O(\eps)) F(\bar{x})$. By
Lemma~\ref{lemma:nearly_fix},
$E[f(D'')] \geq (1- O(\eps)) F(\bar{x})$. Hence, we have

\begin{lemma}
\label{thm:nobig_alg}
For any instance of $\SUB$ with no big elements,
any feasible solution $\bar{x}$ for the continuous problem with $F(\bar{x}) \geq \mO /5$
can be converted to a feasible solution for $\SUB$ in polynomial running time
with expected profit at least $(1-O(\eps))\cdot F(\bar{x})$.
\end{lemma}

\subsection{A Randomized Approximation Algorithm}
\label{sec:mcmb}

Given an instance of $\SUB$ and a subset $T \subseteq U$,
define another instance of $\SUB$, to which we refer as the {\em residual
problem with respect to} $T$, with $f$ remaining the objective function.
The budget for the residual problem is
${\bL}'={\bL}-{\bc}(T)$,
and the universe $U'$ consists of all elements $i \in U \setminus T$
such that $\bc(i) \leq \eps^3 \bL'$,
and all elements in $T$.
Formally,
$$U' = T \cup \left\{ i\in U\setminus T |~ \bc(i) \leq \eps^3 \bL'\right\}.
$$
The new cost of element $i$ is $c'(i)=c(i)$ for any $i\in U' \setminus T$, and
$c'(i) =0 $ for any $i\in T$. It follows that there
are no big elements in the residual problem.
Let $S$ be a feasible solution for the residual problem with respect
to $T$. Then $\bc(S) \leq \bc'(S) +\bc(T) \leq \bL' +\bc(T) =\bL$. Thus,
any feasible solution for the residual problem is also feasible for
the original instance.

Consider the following algorithm.
\\ \\
\noindent
{\bf A Randomized Approximation Algorithm for $\SUB$}\\
\noindent
{\bf Input:} A $\SUB$ instance and an $\alpha$-approximation algorithm
 $\cA$ for continuous $\SUB$ with respect to
the function $f$.
\begin{enumerate}
\item For any $T \subseteq U$ such that $|T| \leq h=\ceil{d \cdot \eps^{-4}}$
\begin{enumerate}
\item
Use $\cA$ to obtain an $\alpha$-approximate solution $\bx$ for the
continuous residual problem with respect to $T$.
\item Use the Randomized Rounding Algorithm of Section~\ref{sec:random_mcmb}
to convert $\bx$ to a feasible solution $S$ for
the residual problem.
\end{enumerate}
\item Return the best solution found.
\end{enumerate}

\begin{lemma}
The above approximation algorithm returns an
$(\alpha -O(\eps))$-approximate solution for $\SUB$
and uses a polynomial number of calls to algorithm $\cA$.
\end{lemma}

\begin{proof}
By Lemma~\ref{thm:nobig_alg},
in each iteration the algorithm finds a feasible solution $S$ for the residual problem.
Hence, the algorithm always
returns a feasible solution for the given $\SUB$ instance.

Let $\mO=\{ i_1,\ldots ,i_k \}$  be an
optimal solution for the input $I$ (we use $\mO$ to denote both an
optimal sub-collection of elements and the optimal value).
For $\ell \geq 1$, let
${\bK}_\ell = \{ i_1,\ldots,i_\ell \}$,
and assume that the elements  are ordered by their residual profits, i.e.,
$i_\ell = \argmax_{i\in \mO \setminus  \bK_{\ell-1}}
f_{\bK_{\ell-1}} (\{i\})$.

Consider the iteration in which $T=\bK_h$,
and define  $\mO'= \mO \cap U'$. The set $\mO'$ is clearly a feasible solution
for the residual problem with respect to $T$.
We show a lower bound for $f(\mO')$.
The set $R=\mO \setminus \mO'$
consists of elements in $\mO \setminus T$ that are big with
respect to the residual instance. The total cost of elements in $R$ is
bounded by $\bL'$ (since $\mO$ is a feasible solution), and thus
$|R| \leq \eps^{-3} \cdot d$.

Since $T=\bK_h$, for any $j \in \mO \setminus T$ it holds that
$f_T(j) \leq \frac{f(T)}{|T|}$, and we get
$f_T(R) \leq \sum_{j\in R} f_T(\{j\}) \leq \eps^{-3} \cdot d \frac{f(T)}{|T|}=
\eps f(T) \leq \eps \mO$.
Thus,  $f_{\mO'} (R) \leq f_T(R) \leq \eps \mO$.
Since $f(\mO)= f(\mO') + f_{\mO'}(R) \leq f(\mO') +\eps f(\mO)$, we have that
$f(\mO') \geq (1-\eps) f(\mO)$.

Thus, in this iteration we get a solution $\bar{x}$ for the residual problem
with $F(\bar{x}) \geq \alpha (1-\eps)f(\mO)$, and the solution $S$ obtained after the
rounding satisfies $f(S) \geq (1-O(\eps)) \alpha f(\mO)$.

\end{proof}
We summarize in the next result.

\begin{theorem}
\label{thm:continuous_eq} Let $f$ be a submodular function, and
suppose there is a polynomial time $\alpha$-approximation
algorithm for the continuous problem with respect to $f$. Then
there is a polynomial time randomized
$(\alpha-\eps)$-approximation algorithm for $\SUB$ with respect to
$f$, for any $\eps>0$.
\end{theorem}
Since there is a $(1/4-o(1))$-approximation algorithm
for general instances of continuous ${\SUB}$
 \cite{LMNS09},
we have

\begin{theorem}
\label{thm:rand_alg_non_monotone}
There is a polynomial time randomized $(1/4-\eps)$-approximation
algorithm for $\SUB$, for any  $\eps>0$.
\end{theorem}
Since there is a $(1-e^{-1}-o(1))$ approximation algorithm for $\SUB$ with monotone objective function \cite{Vo08}
we have

\begin{theorem}
There is a polynomial time randomized $(1-e^{-1}-\eps)$-approximation
algorithm for $\SUB$ with monotone objective function, for any  $\eps>0$.
\end{theorem}


\section{A Deterministic Approximation Algorithm}
\label{mcmb:deter}
In this section we show how the
algorithm of Section \ref{sec:random_mcmb} can be derandomized,
assuming we have
an oracle for $F$, the extension by expectation of $f$.
For some families
of submodular functions, $F$ can
be directly evaluated; for a general function $f$, $F$ can be
evaluated with high accuracy by sampling $f$, as in \cite{Vo08}.

The main idea is to reduce the number of fractional entries in the
fractional solution $\bx$, so that the number of values a random set $D \sim \bx$
can get is polynomial in the input size (for a fixed value of $\eps$).
Then,
we go over all the possible values, and we are promised to obtain a solution
of high value.

A key tool in our derandomization is the \emph{pipage rounding}
technique of Ageev and Sviridenko \cite{as04}. We give below a brief overview
of the technique.
For any element $i\in U$, define the unit vector $\bar{i}\in \{0,1\}^U$, in which
${i}_j=0$
for any $j \neq i$, and ${i}_i=1$.
Given a fractional solution $\bx$ for the problem
and two elements $i,j$, such that $x_i$ and $x_j$ are both
fractional, consider the vector function
$\bx_{i,j}(\delta)= \bx + \delta \bar{i} - \delta \bar{j}$
(Note that $\bx_{i,j}(\delta)$ is equal to $\bx$ in all entries
except $i,j$).
Let $\delta_{\bx,i,j}^{+}$  and  $\delta_{\bx,i,j}^{-}$ (for short, $\delta^{+}$ and  $\delta^{-}$)
be the maximal and minimal value of $\delta$ for which $\bx_{i,j}(\delta) \in [0,1]^U$.
In both $\bx_{i,j}(\delta^{+}), \bx_{i,j}(\delta^{-})$, the entry of either $i$ or $j$ is
integral.

Define  $F^{\bx}_{i,j}(\delta) = F(\bx_{i,j}(\delta))$ over the domain $[\delta^{-},\delta^{+}]$. The
function $F^{\bx}_{i,j}$ is convex (see \cite{ccpv10} for a detailed proof),
thus
$\bx'=\argmax_{\{\bx_{i,j}(\delta^{+}),\bx_{i,j}(\delta^{-})\}} F(\bx)$
has fewer fractional entries than $\bx$, and $F(\bx') \geq F(\bx)$.
By appropriate selection of $i,j$, such that $\bx'$ maintains feasibility (in some sense),
we can repeat the above step
to gradually decrease the number of fractional entries.
We use the technique to prove the next result.

\begin{lemma}
\label{lemma:derand1}
Let $\bx \in [0,1]^U$ be a solution having $k$ or less fractional entries (i.e.,
$|\{i~|~ 0<x_i<1\}| \leq k$),
and $\bc(\bx) \leq \bL$ for some $\bL$. Then $\bx$ can be converted to a vector
$\bx'$ with at most $k'= \left( \frac{8 \ln (2k)}{\eps} \right)^d$
fractional entries,
such that $\bc(\bx')\leq (1+\eps) \bL$,
and $F(\bx') \geq F(\bx)$, in time polynomial in $k$.
\end{lemma}

\begin{proof}
Let $U' = \{ i~ |~ 0 < x_i < 1\}$ be the set of all fractional entries.
We define a new cost function $\bc'$ over the elements in $U$.

\[
\begin{array}{lcr}
c'_r(i) &=& \left\{
\begin{array}{lcl}
c_r(i) &~~~~ &i\notin U' \\
0 & & \displaystyle{ c_r(i) \leq \frac{\eps \cdot L_r}{2k}}\\
\displaystyle{\frac{\eps \cdot L_r}{2k} (1+\eps/2)^j} & &
\displaystyle{\frac{\eps \cdot L_r}{2k} (1+\eps/2)^j \leq c_r(i)<\frac{\eps \cdot L_r}{2k} (1+\eps/2)^{j+1}}
\end{array}
\right.
\end{array}
\]

Note that for any $i\in U'$, $\bc'(i)\leq \bc(i)$, and
$$c_r(i)\leq  \displaystyle{(1+\frac{\eps}{2}) c'_r(i) + \frac{\eps \cdot L_r}{2k}},$$
for all $1 \leq r \leq d$.
The number of different values $c'_r(i)$ can get for $i\in U'$
is bounded by $\frac{8 \ln (2k) }{\eps}$ (since all elements are
small, and $\ln(1+x) \geq x/2$). Hence the number of different
values $\bc'(i)$ can get for $i\in U'$ is bounded by
$k'= \left( \frac{8 \ln (2k) }{\eps} \right) ^d$.

We start with $\bx'=\bx$, and while there are $i,j \in U'$ such that $x'_i$ and $x'_j$ are both
fractional and $\bc'(i)=\bc'(j)$, define $\delta^{+} = \delta_{\bx',i,j}^{+}$
and $\delta^{-}=\delta_{\bx',i,j}^{-}$.
Since $i$ and $j$ have the same cost (by $\bc'$), it holds that
$\bc' \left( \bx_{i,j}(\delta^{+}) \right) =
\bc' \left( \bx_{i,j}(\delta^{-}) \right) =\bc'(\bx)$.
If $F^{\bx}_{i,j}(\delta^{+})\geq F(\bx)$,
then set $\bx''= \bx_{i,j}(\delta^{+})$, otherwise $\bx''= \bx_{i,j}(\delta^{-})$.
In both cases $F(\bx'')\geq F(\bx')$ and $\bc'(\bx'')= \bc'(\bx')$. Now, repeat
this step with $\bx'=\bx''$.
Since in each iteration the number of fractional entries in $\bx'$ decreases,
 the process will terminate (after at most $k$ iterations)
with a vector $\bx'$ such that $F(\bx')\geq F(\bx)$,
$\bc'(\bx')=\bc'(\bx)\leq \bL$,
and there are no two elements $i,j\in U'$ with $\bc'(i)=\bc'(j)$, where
$x'_i$ and $x'_j$ are both fractional. Also, for any $i\notin U'$, the entry
$x'_i$ is integral (since $x_i$ was integral and the entry was not modified by
the process). Thus, the number of fractional entries in $\bx'$ is
at most $k'$.
Now, for any dimension $1 \leq r \leq d$,
\begin{eqnarray*}
c_r(\bx')&=& \sum_{i\notin U'} x'_i c_r(i) + \sum_{i\in U'} x'_i c_r(i) \\
&\leq&
(1+ \eps/2
) \cdot \sum_{i\notin U'} x'_i  \cdot c'_r(i) +
   \sum_{i\in U'} x'_i \left( (1+\eps/2) c'_r(i) + \frac{\eps \cdot L_r}{2k} \right)\\
&=&
(1+\eps /2) \cdot \sum_{i\in U} x'_i  \cdot c'_r(i) +
   \sum_{i\in U'} x_i \frac{\eps \cdot L_r}{2k} \leq (1+\eps) L_r.
\end{eqnarray*}
This completes the proof.
\end{proof}

Using the above lemma, we can reduce the number of fractional
entries in $\bx$ to a number that is poly-logarithmic in $k$.
However, the number of values $D\sim \bx$
 remains super-polynomial. To reduce further the number
of fractional entries, we apply the above step twice,
that is, we convert $\bx$ with at most $|U|$ fractional entries to $\bx'$
with at most
$k'=\left(8 \ln (2 |U| )/\eps \right) ^d$. We
can then
apply the conversion again, to obtain $\bx''$ with at most
$k''= O(\log |U|)$ fractional entries.

\begin{lemma}
\label{lemma:multiple_pipage}
Given a vector $\bL$ and a constant $\eps > 0$,
let $\bx \in [0,1]^U$ be a vector satisfying
$\bc(\bx) \leq \bL$.
Then $\bx$ can be converted in time polynomial in $|U|$ to a vector
$\bx'$ with at most $k''=O(\log |U|)$
  fractional entries, such that $\bc(\bx')\leq (1+\eps)^2 \bL$,
and $F(\bx') \geq F(\bx)$,
\end{lemma}

The next result follows immediately from
Lemma~\ref{lemma:expected_random} ($\mO$ is the value of an optimal solution for ${\SUB}$).

\begin{lemma}
\label{lemma:expected_enumeration}
Given $\bx \in [0,1]^U$ such that $\bx$ is a feasible fractional solution
with $F(\bx) \geq \mO /5$,
there exists
a realization of the random variable $D \sim \bx$, such that the solution
$\cD$ is nearly feasible, and $F(\cD) \geq (1-O(\eps) ) F(\bx)$.
\end{lemma}
}

Consider the following rounding algorithm.\\

\noindent
{\bf Deterministic Rounding Algorithm for $\SUB$ with No Big Elements} \\
\noindent
{\bf Input:} A $\SUB$ instance, a feasible solution  $\bar{x}$ for the continuous problem, with
$F(\bar{x}) \geq \mO /5$.
\begin{enumerate}
\item
Define $\bx'= (1+\eps)^{-2} \cdot \bar{x}$
(note that $F(\bx') \geq (1+\eps)^{-2} \cdot F(\bx)$).
\item Convert $\bx'$ to $\bx''$
such that $\bx''$ is fractionally feasible, the number
of fractional entries in $\bx''$ is $O(\log |U|)$, and
$F(\bx)\geq (1+\eps)^{-2} F(\bx'')\geq (1-e^{-1}-O(\eps)) \mO$,
as in Lemma \ref{lemma:multiple_pipage}.
\item Enumerate over all possible realizations of $D\sim \bx''$.
For each such realization,
if the solution $\cD$ is $\eps$-nearly feasible
convert it to a feasible solution $\cD'$ (see Lemma \ref{lemma:nearly_fix}).
Return the solution with maximum value among the feasible solutions found.
\end{enumerate}

By Theorem~\ref{thm:prb_claim}, the algorithm returns a feasible solution of value at least $(1-O(\eps)) F(\bar{x})$.
Also, the running time of the algorithm is polynomial when $\eps$ is a fixed
constant. Replacing the
randomized rounding step in the
algorithm of Section~\ref{sec:mcmb} with the above Deterministic Rounding Algorithm,
we get the following result.

\begin{theorem}
\label{thm:continuous_det_eq} Let $f$ be a submodular function,
and assume we have an oracle for $F$. If there is a deterministic
polynomial time $\alpha$-approximation algorithm for the
continuous problem with respect to $f$, then there is a polynomial
time deterministic $(\alpha-\eps)$-approximation algorithm for
$\SUB$ with respect to $f$, for any $\eps>0$.
\end{theorem}

We note that,
given an oracle to $F$, both the algorithms of \cite{Vo08} and \cite{LMNS09} for the
continuous problem are deterministic, thus we get
the following.

\begin{theorem}
\label{thm:det_alg}
Given an oracle for $F$,
there is a polynomial time deterministic  $(1-e^{-1} - \eps)$-approximation
algorithm for $\SUB$ with a monotone function, for any $\eps>0$.
\end{theorem}

\begin{theorem}
\label{thm:det_alg_non_monotone}
Given an oracle for $F$,
there is a polynomial time deterministic  $(1/4- \eps)$-approximation
algorithm for $\SUB$ for any $\eps>0$.
\end{theorem}

For the problem of maximum coverage with $d$ knapsack
constraints, i.e.,  $\SUB$ where the objective function is
$f=f_{G,\bar{p}}$, for a given bipartite graph $G$ and profits
$\bar{p}$, the function $F$ can be evaluated deterministically
(see~\cite{as04}). This yields the following result.

\begin{theorem}
There is a polynomial time deterministic
$(1-e^{-1}-\eps)$-approximation algorithm for maximum coverage with
$d$ knapsack constraints.
\end{theorem}

\section{Discussion}

In this paper we established a strong relation between the
continuous relaxation of $\SUB$ and the discrete problem. This
relation is nearly optimal and suggests that future research
should focus on deriving better approximation ratios for the
continuous problem.

The question whether better rounding exists remains open; namely, is it
 possible to obtain an $\alpha-$approximation algorithm for $\SUB$, given
an $\alpha<1$ approximation algorithm for the continuous problem?
And more specifically, is there a polynomial time $(1-e^{-1})-$approximation
for $\SUB$ with monotone objective function?

Finally, the running times of our algorithms are exponential in $1/\eps$, thus rendering them impractical.
 Yet, the hardness results for $d$-dimensional Knapsack (see, e.g., \cite{KPP04,MC84,KS10}),
    a special case of $\SUB$,
  hint that significant improvements over these running times may be impossible.


\bibliographystyle{abbrv}
\bibliography{bgap}

\appendix

\section{Basic Properties of Submodular Functions}
\label{app:basic_props}
In this section we give some simple properties of submodular functions.
Recall that $f : 2^U\rightarrow \mathbb{R}$ is a submodular function
if $f(S) + f(T) \geq f(S \cup T) + f(T \cap S)$ for any $S,T \subseteq U$.
We define $f_T(S) = f(S\cup T) - f(T)$.
\begin{lemma}
\label{lemma:submodular_summation}
Let $f: 2^U\rightarrow \mathbb{R}$ be a submodular function with $f(\emptyset)\geq 0$,
and let $S= S_1 \cup S_2 \cup \ldots \cup S_k$, where $S_i$ are disjoint sets.
Then $$f(S) \geq f(S_1) +f(S_2)+\ldots f(S_k).$$
\end{lemma}
\begin{proof}
By induction on $k$. For $k=2$, since $f$ is a submodular function, we have that
$$f(S_1) + f(S_2) \geq f(S_1 \cup S_2) + f(S_1 \cap S_2)= f(S) +f(\emptyset),$$
and since $f(\emptyset) \geq 0$, we get that $f(S) \leq f(S_1) + f(S_2)$.

For $k>2$, using the induction hypothesis twice, we have
$$f(S)\leq f(S_1)+f(S_2)+\ldots f(S_{k-2})  + f(S_{k-1}\cup S_{k}) \leq
f(S_1) +f(S_2)+\ldots f(S_k).$$
\end{proof}
\begin{lemma}
\label{lemma:submodular_decreasing}
Let $f: 2^U\rightarrow \mathbb{R}_+$ be a submodular function, and let $S,T_1,T_2 \subseteq U$
such that $T_1 \subseteq T_2$ and $S \cap T_2 =\emptyset$. Then,
$f_{T_2}(S) \leq f_{T_1}(S)$.
\end{lemma}
\begin{proof}
Since $f$ is submodular,
$$f(S\cup T_1) +f(T_2) \geq f(S\cup T_1 \cup T_2) + f( (S \cup T_1)\cap T_2) =
f(S \cup T_2) +f(T_1).$$
Hence, $f_{T_2}(S) \leq f_{T_1}(S)$.
\end{proof}
\begin{lemma}
\label{lemma:submodular_summation_inequality}
Let $f: 2^U\rightarrow \mathbb{R}_+$ be a submodular function, and let
$S= S_1 \cup S_2 \cup \ldots \cup S_k$,
where $S_i$ are disjoint sets. Then,
$$f(S) \geq \sum_{i=1}^{k} f_{S\setminus S_i} (S_i).$$
\end{lemma}
\begin{proof}
We note that
$$f(S) = \sum_{i=1}^{k} f_{S_1 \cup \ldots \cup S_{i-1}} (S_i).$$
By Lemma~\ref{lemma:submodular_decreasing}, for each $i >1$,
$f_{S_1 \cup \ldots \cup S_{i-1}} (S_i) \geq  f_{S\setminus S_i} (S_i)$.
Hence,
$$f(S) \geq \sum_{i=1}^{k} f_{S\setminus S_i} (S_i).$$
\end{proof}

\end{document}